\documentclass[12pt]{article}

\usepackage{algorithm}              % algorithms
\usepackage[noend]{algpseudocode}   % algorithms
\usepackage{amsmath}                % math package
\usepackage{amssymb}                % math package
\usepackage{amsthm}                 % math package
\usepackage{caption}
\usepackage{color}
\usepackage[margin=1in]{geometry}
\usepackage[hidelinks]{hyperref}
\usepackage{titling}

\usepackage{forest} % Fibonacci's word recursion tree

\forestset{
  spine edge/.style={
    edge={draw=black, very thick},
  },
  dotted spine edge/.style={
    edge={draw=black, ultra thick, dotted},
  },
  factor edge/.style={
    edge={draw=black, thick, dashed},
  },
  dashdotted factor edge/.style={
    for tree={
      text=gray!70,
      edge={draw=gray!55, thin},
    },
    edge={draw=black, thick, dash dot},
  },
  context subtree/.style={
    for tree={
      text=gray!70,
      edge={draw=gray!55, thin},
    },
  },
}

\theoremstyle{plain}                % AMS theorem styles
\newtheorem{theorem}{Theorem}
\newtheorem{lemma}[theorem]{Lemma}

\newtheorem{conjecture}[theorem]{Conjecture}

\theoremstyle{definition}           % AMS definition styles

\newtheorem{example}[theorem]{Example}

\theoremstyle{remark}               % AMS remark styles
\newtheorem*{remark}{Remark}

\newtheorem*{claim}{Claim}

\thanksmarkseries{alph}

\newcommand{\Fib}[1]{F_{#1}}
\newcommand{\from}{\colon}
\newcommand{\PD}{\operatorname{PD}}
\newcommand{\PSD}{\operatorname{PSD}}
\newcommand{\SD}{\operatorname{SD}}
\newcommand{\SSC}{\operatorname{SSC}}

\usepackage{stackengine}
\newcommand{\rectangle}{{\ooalign{$\sqsubset\mkern6mu$\cr$\mkern6mu\sqsupset$\cr}}}
\newcommand{\PSSR}{\stackinset{c}{}{c}{}{\rule{.6pt}{2ex}}{$\rectangle$}}

\title{Generating Fibonacci Words via the Prefix--Suffix Duplication Operation}
\author{
    Diego Cabrera Salamanca \thanks{Departamento de Matem\'{a}ticas, Universidad Nacional de Colombia, Bogot\'{a}, Colombia. Email: \href{mailto:dicabreras@unal.edu.co}{\texttt{dicabreras@unal.edu.co}}.}
    \and
    Taylor J. Smith \thanks{Department of Computer Science, St.\ Francis Xavier University, Antigonish, Nova Scotia, Canada. Email: \href{mailto:tjsmith@stfx.ca}{\texttt{tjsmith@stfx.ca}}.}
}
\date{\today}

\begin{document}

%%%%%%

\maketitle

\begin{abstract}
The finite and infinite Fibonacci words are classical objects in combinatorics on words. Bio-inspired language operations provide a useful tool for studying how finite and infinite words can arise via local rewriting mechanisms. For example, the suffix square completion operation is known to generate the infinite Fibonacci word, as well as other infinite words such as the Thue–-Morse word and the period-doubling word.

The prefix-–suffix duplication operation produces a language of words formed by appending prefixes or suffixes of a word $w$ to the front or back of $w$ respectively, and Dumitran conjectured that Fibonacci words of the same index parity can be generated from one another by the bounded duplication length variant of this operation.

In this paper, we resolve and strengthen Dumitran's conjecture. We show that, for all $1 \leq p \leq n$, it is possible to generate the Fibonacci word $\Fib{2n}$ from $\Fib{2p}$, and $\Fib{2n+1}$ from $\Fib{2p+1}$, using only prefix duplications with a bound of $k \geq 3$. We furthermore show that this bound of $3$ is optimal, and we give an algorithm that produces a witness derivation in time linear in the length of the target Fibonacci word.

\medskip

\noindent\textit{Key words and phrases:} Bio-inspired language operation, duplication operation, Fibonacci word, prefix--suffix duplication

\medskip

\noindent\textit{MSC2020 classes:} 11B39 (primary); 68Q45, 68R15 (secondary).
\end{abstract}

%%%%%%

\section{Introduction}\label{sec:introduction}

The well-studied sequence of finite Fibonacci words can be defined either by the two-symbol alphabet morphism $\phi(0) = 01$ and $\phi(1) = 0$, or by the recurrence $\Fib{i} = \Fib{i-1}\Fib{i-2}$ with base cases $\Fib{0} = 0$ and $\Fib{1} = 01$. By iterating either the morphism or the recurrence, we obtain the infinite Fibonacci word $\mathbf{F} = 0100101001001\dots$.

It is known~\cite{Dumitran2015PrefixSuffixSquareCompletion} that the infinite Fibonacci word can be generated by the bio-inspired \emph{suffix square completion} operation, which is defined by
\begin{equation*}
\SSC(w) = \{wx \mid w = w'yxy \text{ for some } x \in \Sigma^{+} \text{ and } w', y \in \Sigma^{*}\}.
\end{equation*}
Other well-known infinite words, such as the period-doubling word and the Thue--Morse word, can be generated by this same operation. Indeed, this gives rise to a fascinating observation: applying the suffix square completion operation once to a word necessarily generates a word that contains a square, yet iterated application of the suffix square completion operation can generate infinite words that avoid any word power of rational exponent strictly greater than $2$.

By contrast, the infinite Thue--Morse word cannot be generated by the bio-inspired \emph{prefix--suffix duplication} operation $\PSD(w)$, which produces a language of words formed by appending prefixes or suffixes of a word $w$ to the front or back of $w$, respectively. A similar result is not known for the infinite Fibonacci word, but by defining a \emph{bounded} variant of the prefix--suffix duplication operation $\PSD_{k}(w)$, where appended prefixes and suffixes are constrained to have length at most $k$, it is possible that the Fibonacci word and other infinite words constructed via concatenation may be generated in this way.

Indeed, in his thesis~\cite{Dumitran2015PhDThesis}, Dumitran states the following conjecture:
\begin{conjecture}\label{conj:dumitran}
$\Fib{2n} \in \PSD_{k}^{*}(\Fib{2p})$ for any $1 \leq p \leq n$ and $k \geq 3$, and $\Fib{2n+1} \in \PSD_{k}^{*}(\Fib{2p+1})$ for any $1 \leq p \leq n$ and $k \geq 5$.
\end{conjecture}
In other words, Dumitran conjectures that we may obtain the infinite Fibonacci word by iteratively applying the bounded prefix--suffix duplication operation, but we may only obtain even-indexed Fibonacci words from even-indexed Fibonacci words with a bounded length of at least $3$ and likewise for odd indices with a bounded length of at least $5$.

In this paper, we resolve Dumitran's conjecture in the positive, demonstrating that Fibonacci words of the same index parity can be obtained from one another using bounded prefix--suffix duplications. In fact, we prove a stronger result demonstrating that only prefix duplications are necessary, and where the lengths of each duplication are bounded by $k \geq 3$ in both the even and odd cases. We furthermore show that the bound $k \geq 3$ is optimal. Consequently, arbitrarily long prefixes of the infinite Fibonacci word can be obtained from shorter same-parity Fibonacci words via a finite number of prefix duplications.

%%%%%%

\section{Preliminaries}\label{sec:prelims}

In this section, we define the main notions considered in this paper. For further details about formal language theory and combinatorics on words, see (for example) the survey by Choffrut and Karhum\"{a}ki~\cite{ChoffrutKarhumaki1997Combinatorics} or the book by Lothaire~\cite{Lothaire1983CombinatoricsOnWords}.

%%%

\subsection{Fibonacci Words}\label{subsec:fibonacci}

Our primary objects of interest in this paper are the finite and infinite \emph{Fibonacci words}. Given a binary alphabet $\Sigma = \{0, 1\}$, we define the Fibonacci morphism $\phi \from \Sigma^{*} \to \Sigma^{*}$ by
\begin{align*}
\phi(0) &= 01 \text{ and} \\
\phi(1) &= 0.
\end{align*}
We can then iterate this morphism to produce the sequence of finite Fibonacci words:
\begin{equation*}
\Fib{0} = 0, \Fib{1} = 01, \Fib{2} = 010, \Fib{3} = 01001, \Fib{4} = 01001010, \dots
\end{equation*}
Alternatively, we can obtain the $i$th Fibonacci word by starting with $\Fib{0} = 0$ and $\Fib{1} = 01$, then concatenating the two previous Fibonacci words to produce $\Fib{i} = \Fib{i-1}\Fib{i-2}$. This construction mimics how we obtain the $i$th Fibonacci number by adding together the two previous Fibonacci numbers.

We can also produce the infinite Fibonacci word $\mathbf{F} = 0100101001001\dots$ by taking the limit of the sequence $(\Fib{n})_{n \geq 0}$. Fibonacci words (as we have defined them here) seem to have been introduced by Knuth~\cite[Section 1.2.8, exercise 36]{Knuth1968TAOCPv1}; such words and their generation by various means were further studied by Chuan~\cite{Chuan1992FibonacciWords, Chuan1995GeneratingFibonacciWords}. For more details about Fibonacci words, see the survey by Berstel~\cite{Berstel1986FibonacciWords}.

%%%

\subsection{Duplication Operations}\label{subsec:duplication}

A number of operations in formal language theory have their roots in biology, and are therefore termed \emph{bio-inspired} language operations. Here, we consider three particular operations. Given a word $w$ over an alphabet $\Sigma$, we may define the following \emph{prefix--suffix duplication} operations:
\begin{itemize}
\item prefix duplication, $\PD(w) = \{uw \mid w = uy \text{ for some } u \in \Sigma^{+} \text{ and } y \in \Sigma^{*}\}$;
\item suffix duplication, $\SD(w) = \{wu \mid w = yu \text{ for some } u \in \Sigma^{+} \text{ and } y \in \Sigma^{*}\}$; and
\item prefix--suffix duplication, $\PSD(w) = \PD(w) \cup \SD(w)$.
\end{itemize}

\begin{example}
Suppose $w = 0101$. The nonempty prefixes of $w$ are $0$, $01$, $010$, and $0101$, so we have that
\begin{equation*}
\PD(w) = \{00101, 010101, 0100101, 01010101\}.
\end{equation*}
Likewise, the nonempty suffixes of $w$ are $1$, $01$, $101$, and $0101$, so we have that
\begin{equation*}
\SD(w) = \{01011, 010101, 0101101, 01010101\}.
\end{equation*}
Altogether, this produces
\begin{equation*}
\PSD(w) = \{00101, 01011, 010101, 0100101, 0101101, 01010101\}.
\end{equation*}
\end{example}

Additionally, given a positive integer $k$, we can further define \emph{bounded} variants of each prefix--suffix duplication operation:
\begin{itemize}
\item $k$-prefix duplication, $\PD_{k}(w) = \{uw \mid w = uy \text{ for some } u \in \Sigma^{+}, y \in \Sigma^{*} \text{ and } |u| \leq k\}$;
\item $k$-suffix duplication, $\SD_{k}(w) = \{wu \mid w = yu \text{ for some } u \in \Sigma^{+}, y \in \Sigma^{*} \text{ and } |u| \leq k\}$; and
\item $k$-prefix--suffix duplication, $\PSD_{k}(w) = \PD_{k}(w) \cup \SD_{k}(w)$.
\end{itemize}
It is straightforward to see that $\PD_{k}(w) \subseteq \PD(w)$, $\SD_{k}(w) \subseteq \SD(w)$, and $\PSD_{k}(w) \subseteq \PSD(w)$ for all $k$, and that $\PD_{k}(w) \subseteq \PD_{\ell}(w)$, $\SD_{k}(w) \subseteq \SD_{\ell}(w)$, and $\PSD_{k}(w) \subseteq \PSD_{\ell}(w)$ for all $k \leq \ell$.

\begin{example}
Suppose again that $w = 0101$, and let $k = 2$. In this case, we may only duplicate the prefixes $0$ and $01$ or the suffixes $1$ and $01$, and so we have that
\begin{equation*}
\PSD_{2}(w) = \{00101, 01011, 010101\}.
\end{equation*}
\end{example}

Lastly, we can iterate duplication operations in the usual way. For any duplication operation $\mathrm{Op} \in \{\PD, \SD, \PSD, \PD_{k}, \SD_{k}, \PSD_{k}\}$, we define
\begin{align*}
\mathrm{Op}^{0}(w)      &= \{w\}; \\
\mathrm{Op}^{i+1}(w)    &= \mathrm{Op}^{i}(w) \cup \mathrm{Op}(\mathrm{Op}^{i}(w)); \text{ and} \\
\mathrm{Op}^{*}(w)      &= \bigcup_{i \geq 0} \mathrm{Op}^{i}(w).
\end{align*}
The prefix--suffix duplication operation was introduced by Garc\'{\i}a Lopez, Manea, and Mitrana~\cite{GarciaLopez2014PrefixSuffixDuplication}, while the bounded variant of the operation was introduced by Dumitran, Gil, Manea, and Mitrana~\cite{Dumitran2015BoundedPrefixSuffixDuplication}.

%%%%%%

\section{Proof of Dumitran's Conjecture}\label{sec:proof}

In this section, we provide a proof resolving Dumitran's conjecture in the positive. Consequently, arbitrarily long prefixes of the infinite Fibonacci word can be obtained from shorter same-parity Fibonacci words via a finite number of prefix duplications.

We begin by establishing two elementary decompositions of Fibonacci words. The first decomposition shows that, by repeatedly expanding the even-indexed right factor in the Fibonacci recurrence, any even-indexed Fibonacci word $F_{2n}$ can be written as a concatenation of decreasing odd-indexed Fibonacci words followed by a smaller even-indexed Fibonacci word.

\begin{lemma}\label{lem:evenconcat}
For every $n \geq 1$ and every $1 \leq i \leq n$, we may write $\Fib{2n}$ as a decomposition
\begin{equation*}
\Fib{2n} = \Fib{2n-1}\Fib{2n-3} \cdots \Fib{2n-2i+1}\Fib{2n-2i}.
\end{equation*}
\end{lemma}

\begin{proof}
Fix $n \geq 1$. We prove the claim by induction on $i$.

For $i = 1$, the assertion is exactly the Fibonacci recurrence $\Fib{2n} = \Fib{2n-1}\Fib{2n-2}$.

Now, suppose the formula holds for some $i < n$. Then we may write
\begin{equation*}
\Fib{2n} = \Fib{2n-1}\Fib{2n-3} \cdots \Fib{2n-2i+1}\Fib{2n-2i}.
\end{equation*}
Since $i < n$, we have $2n-2i \geq 2$, and hence the recurrence gives $\Fib{2n-2i} = \Fib{2n-2i-1}\Fib{2n-2i-2}$. Substituting this into the previous expression for $\Fib{2n}$ gives
\begin{align*}
\Fib{2n}	&= \Fib{2n-1}\Fib{2n-3} \cdots \Fib{2n-2i+1}\Fib{2n-2i-1}\Fib{2n-2i-2} \\
		&= \Fib{2n-1}\Fib{2n-3} \cdots \Fib{2n-2(i+1)+3}\Fib{2n-2(i+1)+1}\Fib{2n-2(i+1)},
\end{align*}
which is the desired formula with $i + 1$ in place of $i$. Therefore, the formula holds for every $1 \leq i \leq n$.
\end{proof}

The second decomposition has an analogous proof. By repeatedly expanding the odd-indexed right factor in the Fibonacci recurrence, any odd-indexed Fibonacci word $F_{2n+1}$ can be written as a concatenation of decreasing even-indexed Fibonacci words followed by a smaller odd-indexed Fibonacci word.

\begin{lemma}\label{lem:oddconcat}
For every $n \geq 1$ and every $1 \leq i \leq n$, we may write $\Fib{2n+1}$ as a decomposition
\begin{equation*}
\Fib{2n+1} = \Fib{2n}\Fib{2n-2}\Fib{2n-4} \cdots \Fib{2n-2i+2}\Fib{2n-2i+1}.
\end{equation*}
\end{lemma}

\begin{proof}
Fix $n \geq 1$. We prove the claim by induction on $i$.

For $i = 1$, the formula is again exactly the Fibonacci recurrence $\Fib{2n+1} = \Fib{2n}\Fib{2n-1}$.

Now suppose the formula holds for some $i < n$. Then we may write
\begin{equation*}
\Fib{2n+1} = \Fib{2n}\Fib{2n-2}\Fib{2n-4} \cdots \Fib{2n-2i+2}\Fib{2n-2i+1}.
\end{equation*}
Since $i < n$, we have $2n-2i+1 \geq 3$, and hence the recurrence gives $\Fib{2n-2i+1} = \Fib{2n-2i}\Fib{2n-2i-1}$. Substituting this into the previous expression for $\Fib{2n+1}$ gives
\begin{align*}
\Fib{2n+1}	&= \Fib{2n}\Fib{2n-2}\Fib{2n-4} \cdots \Fib{2n-2i+2}\Fib{2n-2i}\Fib{2n-2i-1} \\
			&= \Fib{2n}\Fib{2n-2}\Fib{2n-4} \cdots \Fib{2n-2(i+1)+4}\Fib{2n-2(i+1)+2}\Fib{2n-2(i+1)+1},
\end{align*}
which is the desired formula with $i + 1$ in place of $i$. Therefore, the formula holds for every $1 \leq i \leq n$.
\end{proof}

We will often require the fact that the current Fibonacci word we are working with is prefixed by the Fibonacci word $\Fib{2}$. The following lemma proves that this is the case for all Fibonacci words aside from the base cases $\Fib{0}$ and $\Fib{1}$.

\begin{lemma}\label{lem:F2prefix}
For all $n \geq 2$, $\Fib{2}$ is a prefix of $\Fib{n}$.
\end{lemma}

\begin{proof}
If $n = 2$, then the claim is immediate. Now, suppose $n \geq 3$ and $\Fib{2}$ is a prefix of $\Fib{n-1}$. By the Fibonacci recurrence,
\begin{equation*}
\Fib{n} = \Fib{n-1}\Fib{n-2}.
\end{equation*}
Since $\Fib{n-1}$ is the first factor of this concatenation, every prefix of $\Fib{n-1}$ is also a prefix of $\Fib{n}$; in particular, $\Fib{2}$ is a prefix of $\Fib{n}$. Thus, the claim holds by induction.
\end{proof}

\begin{remark}
In what follows, we will use shorthand notation inspired by formal language theory and write $u \Rightarrow_{k} v$ when $v \in \PD_{k}(u)$. Similarly, we will write $u \Rightarrow_{k}^{*} v$ when $v \in \PD_{k}^{*}(u)$.
\end{remark}

Another useful fact we will require is that applying the prefix duplication operation to a word $u$ to obtain a word $v$ produces the same result even when $u$ and $v$ are each suffixed by an arbitrary word $w$.

\begin{lemma}\label{lem:rightcontext}
If $u \Rightarrow_{k}^{*} v$, then $uw \Rightarrow_{k}^{*} vw$ for every word $w$.
\end{lemma}

\begin{proof}
It is sufficient to prove the statement for one prefix duplication step. Suppose $u \Rightarrow_{k} v$. Then $v = pu$, where $p$ is a nonempty prefix of $u$ and $|p| \leq k$. But then $p$ is also a nonempty prefix of $uw$, again with $|p| \leq k$. Therefore, $uw \Rightarrow_{k} puw = vw$. Applying this argument to each step in a derivation from $u$ to $v$ gives $uw \Rightarrow_{k}^{*} vw$.
\end{proof}

In other terms, this lemma allows us to derive one Fibonacci word from another while keeping the remainder of the current word unchanged on the right.

%%%

\subsection{Main Results}\label{subsec:mainresults}

Now, we turn our attention to the main machinery of our proof. To begin, we will show that every even-indexed Fibonacci word can be derived from the Fibonacci word $\Fib{2}$ via the prefix duplication operation. Since Lemma~\ref{lem:F2prefix} tells us that every Fibonacci word is prefixed by $\Fib{2}$, this result will give us a starting point from which we will be able to produce arbitrary prefixes in our main theorems.

\begin{lemma}\label{lem:evenPD}
$\Fib{2n} \in \PD_{k}^{*}(\Fib{2})$ for all $n \geq 1$ and all $k \geq 3$.
\end{lemma}

\begin{proof}
We prove the result by strong induction on $n$.

For $n = 1$, we have $\Fib{2n} = \Fib{2}$, so $\Fib{2} \in \PD_{k}^{*}(\Fib{2})$ by definition.

Now, assume that $n \geq 1$ and that $\Fib{2j} \in \PD_{k}^{*}(\Fib{2})$ for every $1 \leq j \leq n$. We prove that $\Fib{2n+2} \in \PD_{k}^{*}(\Fib{2})$. By the recurrence, we know that $\Fib{2n+2} = \Fib{2n+1}\Fib{2n}$. Then, by applying Lemma~\ref{lem:oddconcat} with $i = n$, we obtain $\Fib{2n+1} = \Fib{2n}\Fib{2n-2} \cdots \Fib{2}\Fib{1}$. Altogether, we may write
\begin{equation*}
\Fib{2n+2} = \Fib{2n}\Fib{2n-2} \cdots \Fib{2}\Fib{1}\Fib{2n}.
\end{equation*}

By the strong induction hypothesis with $j = n$, we know that $\Fib{2} \Rightarrow_{k}^{*} \Fib{2n}$. Thus, we may derive $\Fib{2n}$ from $\Fib{2}$. Since $\Fib{2n}$ has prefix $\Fib{2}$ by Lemma~\ref{lem:F2prefix}, it also has prefix $\Fib{1}$, because $\Fib{2} = 010$ and $\Fib{1} = 01$. Moreover, $|\Fib{1}| = 2 \leq k$. Hence, one prefix duplication step produces
\begin{equation*}
\Fib{2n} \Rightarrow_{k} \Fib{1}\Fib{2n}.
\end{equation*}
The word $\Fib{1}\Fib{2n}$ again has prefix $\Fib{2}$: indeed, $\Fib{2n}$ begins with $010$, so $\Fib{1}\Fib{2n} = 01010\dots$, which begins with $\Fib{2} = 010$.

We now prepend, one at a time, the words $\Fib{2}$, $\Fib{4}$, $\Fib{6}$, \dots, $\Fib{2n}$.
Suppose, at some stage, the current word is $W$ and $W$ has prefix $\Fib{2}$. Let $1 \leq j \leq n$. Since $W$ has prefix $\Fib{2}$, and since $|\Fib{2}| = 3 \leq k$, one prefix duplication step produces
\begin{equation*}
W \Rightarrow_{k} \Fib{2}W.
\end{equation*}
By the induction hypothesis, $\Fib{2} \Rightarrow_{k}^{*} \Fib{2j}$. Then, by Lemma~\ref{lem:rightcontext}, we know that
\begin{equation*}
\Fib{2}W \Rightarrow_{k}^{*} \Fib{2j}W.
\end{equation*}
Therefore, $W \Rightarrow_{k}^{*} \Fib{2j}W$, and $\Fib{2j}W$ has prefix $\Fib{2}$ again by Lemma~\ref{lem:F2prefix}.

Starting from $\Fib{1}\Fib{2n}$, apply this argument successively with $j = 1, 2, \ldots, n$.
This produces the derivation
\begin{equation*}
\Fib{1}\Fib{2n} \Rightarrow_{k}^{*} \Fib{2}\Fib{1}\Fib{2n} \Rightarrow_{k}^{*} \Fib{4}\Fib{2}\Fib{1}\Fib{2n} \Rightarrow_{k}^{*} \cdots \Rightarrow_{k}^{*} \Fib{2n}\Fib{2n-2} \cdots \Fib{2}\Fib{1}\Fib{2n}.
\end{equation*}
Observe that the word produced by this derivation sequence is exactly $\Fib{2n+2}$. Hence, $\Fib{2n+2} \in \PD_{k}^*(\Fib{2})$ as required.
\end{proof}

We will find the following consequence of Lemma~\ref{lem:evenPD} especially useful in the proofs of our main theorems. Suppose $W$ has prefix $\Fib{2}$, and let $r \geq 1$. By Lemma~\ref{lem:evenPD}, $\Fib{2} \Rightarrow_{k}^{*} \Fib{2r}$. Since $W$ has prefix $\Fib{2}$, we first duplicate this prefix to obtain $W \Rightarrow_{k} \Fib{2}W$. Then, again by Lemma~\ref{lem:rightcontext}, we know that $\Fib{2}W \Rightarrow_{k}^{*} \Fib{2r}W$, so $W \Rightarrow_{k}^{*} \Fib{2r}W$. Thus, whenever the current word begins with $\Fib{2}$, we may prepend any even-indexed Fibonacci word $\Fib{2r}$ to the current word.

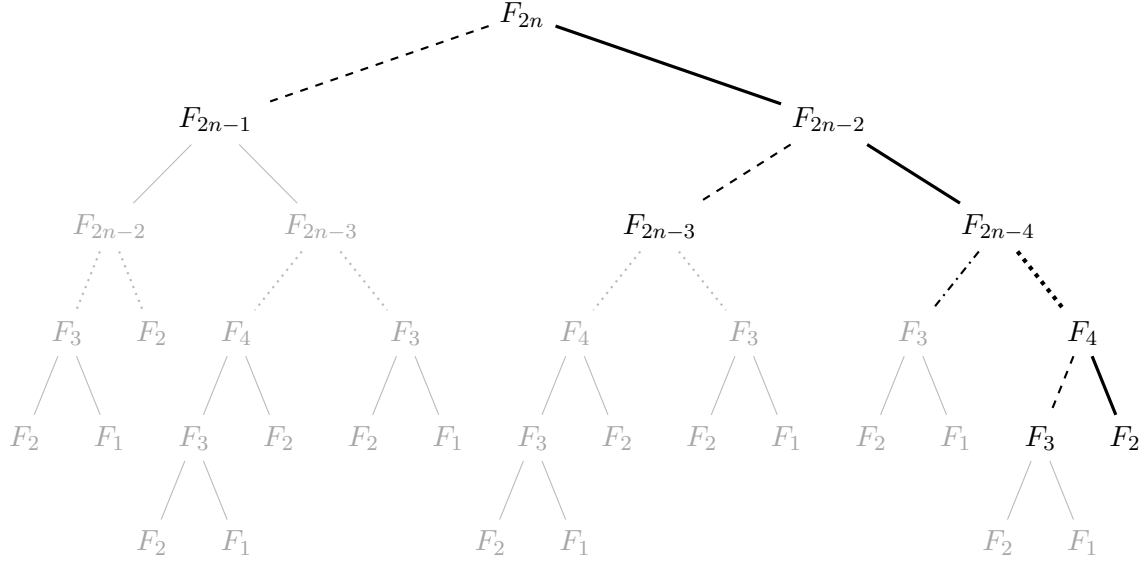
\begin{figure}[t]
\centering
\begin{forest}
for tree={
  font=\small,
  l sep+=8pt,
  s sep+=4pt,
}
[$\Fib{2n}$
  [$\Fib{2n-1}$, factor edge
    [$\Fib{2n-2}$, context subtree
      [$\Fib{3}$, edge={thick, dotted}
        [$\Fib{2}$]
        [$\Fib{1}$]
      ]
      [$\Fib{2}$, edge={thick, dotted}]
    ]
    [$\Fib{2n-3}$, context subtree
      [$\Fib{4}$, edge={thick, dotted}
        [$\Fib{3}$ [$\Fib{2}$] [$\Fib{1}$]]
        [$\Fib{2}$]
      ]
      [$\Fib{3}$, edge={thick, dotted}
        [$\Fib{2}$]
        [$\Fib{1}$]
      ]
    ]
  ]
  [$\Fib{2n-2}$, spine edge
    [$\Fib{2n-3}$, factor edge
      [$\Fib{4}$, context subtree, edge={thick, dotted}
        [$\Fib{3}$ [$\Fib{2}$] [$\Fib{1}$]]
        [$\Fib{2}$]
      ]
      [$\Fib{3}$, context subtree, edge={thick, dotted}
        [$\Fib{2}$]
        [$\Fib{1}$]
      ]
    ]
    [$\Fib{2n-4}$, spine edge
      [$\Fib{3}$, dashdotted factor edge
        [$\Fib{2}$, context subtree]
        [$\Fib{1}$, context subtree]
      ]
      [$\Fib{4}$, dotted spine edge
        [$\Fib{3}$, factor edge
            [$\Fib{2}$, context subtree]
            [$\Fib{1}$, context subtree]
        ]
        [$\Fib{2}$, spine edge]
      ]
    ]
  ]
]
\end{forest}
\caption{Recursive decomposition of $\Fib{2n}$ from Lemma~\ref{lem:evenconcat} underlying the construction in the proof of Theorem~\ref{thm:even}. Dashed edges indicate the odd-indexed Fibonacci word factors exposed by recursively expanding the even-indexed Fibonacci words along the thick edges.}
\label{fig:tree-even}
\end{figure}

We may now prove the even-indexed case of Dumitran's conjecture. Interestingly, while Dumitran's conjecture uses the prefix--suffix duplication operation $\PSD_{k}^{*}$, our approach only uses prefix duplication and allows us to disregard suffixes entirely.

\begin{theorem}\label{thm:even}
$\Fib{2n} \in \PD_{k}^{*}(\Fib{2p})$ for all $1 \leq p \leq n$ and all $k \geq 3$.
\end{theorem}

\begin{proof}
If $p = n$, then the claim is immediate, since $\Fib{2n} \in \PD_{k}^{*}(\Fib{2n})$.

Now, assume that $p < n$. Applying Lemma~\ref{lem:evenconcat} with $i = n-p$ produces
\begin{equation*}
\Fib{2n} = \Fib{2n-1}\Fib{2n-3} \cdots \Fib{2p+1}\Fib{2p}.
\end{equation*}
Thus, it is enough to show that starting from $\Fib{2p}$, we may successively prepend the odd-indexed Fibonacci words $\Fib{2p+1}$, $\Fib{2p+3}$, \dots, $\Fib{2n-1}$.

We first show how to prepend a single odd-indexed Fibonacci word to any current word beginning with $\Fib{2}$. 

\begin{claim}
Let $W$ be a word with prefix $\Fib{2}$, and let $q \geq 1$. Then
\begin{equation*}
W \Rightarrow_{k}^{*} \Fib{2q+1}W.
\end{equation*}
\end{claim}

\begin{proof}[Proof of claim]
Since $W$ has prefix $\Fib{2} = 010$, it has prefix $\Fib{1} = 01$. Since $|\Fib{1}| = 2 \leq k$, one prefix duplication step produces $W \Rightarrow_{k} \Fib{1}W$. The word $\Fib{1}W$ again begins with $\Fib{2}$, because $W$ begins with $010$, so $\Fib{1}W = 01010\dots$.

Applying Lemma~\ref{lem:oddconcat} with $i = q$ produces
\begin{equation*}
\Fib{2q+1} = \Fib{2q}\Fib{2q-2} \cdots \Fib{2}\Fib{1}.
\end{equation*}
By our earlier observed consequence of Lemma~\ref{lem:evenPD}, we may now successively prepend $\Fib{2}$, $\Fib{4}$, \dots, $\Fib{2q}$ to the current word. This produces the derivation
\begin{equation*}
\Fib{1}W \Rightarrow_{k}^{*} \Fib{2}\Fib{1}W \Rightarrow_{k}^{*} \Fib{4}\Fib{2}\Fib{1}W \Rightarrow_{k}^{*} \cdots \Rightarrow_{k}^{*} \Fib{2q}\Fib{2q-2} \cdots \Fib{2}\Fib{1}W.
\end{equation*}
Observe that the word produced by this derivation sequence is $\Fib{2q+1}W$. Hence, $W \Rightarrow_{k}^{*} \Fib{2q+1}W$ whenever $W$ has the prefix $\Fib{2}$.
\end{proof}

Now, we return to the proof of the theorem. Since $2p \geq 2$, Lemma~\ref{lem:F2prefix} says that $\Fib{2p}$ has prefix $\Fib{2}$. Applying our claim with $W = \Fib{2p}$ and $q = p$, we obtain
\begin{equation*}
\Fib{2p} \Rightarrow_{k}^{*} \Fib{2p+1}\Fib{2p}.
\end{equation*}
The resulting word begins with $\Fib{2}$, since $\Fib{2p+1}$ begins with $\Fib{2}$. Applying the same argument again with $q = p+1$, then with $q = p+2$, and so on, gives
\begin{equation*}
\Fib{2p} \Rightarrow_{k}^{*} \Fib{2p+1}\Fib{2p} \Rightarrow_{k}^{*} \Fib{2p+3}\Fib{2p+1}\Fib{2p} \Rightarrow_{k}^{*} \cdots \Rightarrow_{k}^{*} \Fib{2n-1}\Fib{2n-3} \cdots \Fib{2p+1}\Fib{2p}.
\end{equation*}
Observe that the word produced by this derivation sequence is $\Fib{2n}$. Hence, $\Fib{2n} \in \PD_{k}^{*}(\Fib{2p})$ as required.
\end{proof}

Naturally, since $\PD_{k}(w) \subseteq \PSD_{k}(w)$, Theorem~\ref{thm:even} implies the even-indexed case of Dumitran's conjecture.

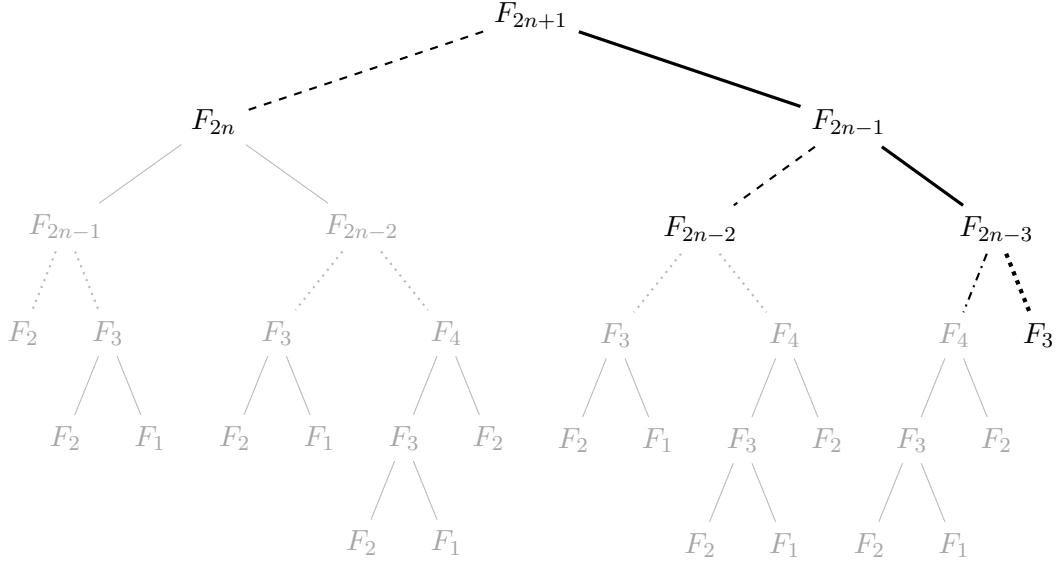
\begin{figure}[t]
\centering
\begin{forest}
for tree={
  font=\small,
  l sep+=8pt,
  s sep+=4pt,
}
[$\Fib{2n+1}$
  [$\Fib{2n}$, factor edge
    [$\Fib{2n-1}$, context subtree
      [$\Fib{2}$, edge={thick, dotted}]
      [$\Fib{3}$, edge={thick, dotted}
        [$\Fib{2}$]
        [$\Fib{1}$]
      ]
    ]
    [$\Fib{2n-2}$, context subtree
      [$\Fib{3}$, edge={thick, dotted}
        [$\Fib{2}$]
        [$\Fib{1}$]
      ]
      [$\Fib{4}$, edge={thick, dotted}
        [$\Fib{3}$ [$\Fib{2}$] [$\Fib{1}$]]
        [$\Fib{2}$]
      ]
    ]
  ]
  [$\Fib{2n-1}$, spine edge
    [$\Fib{2n-2}$, factor edge
      [$\Fib{3}$, context subtree, edge={thick, dotted}
        [$\Fib{2}$]
        [$\Fib{1}$]
      ]
      [$\Fib{4}$, context subtree, edge={thick, dotted}
        [$\Fib{3}$
            [$\Fib{2}$]
            [$\Fib{1}$]
        ]
        [$\Fib{2}$]
      ]
    ]
    [$\Fib{2n-3}$, spine edge
      [$\Fib{4}$, dashdotted factor edge
        [$\Fib{3}$, context subtree
            [$\Fib{2}$]
            [$\Fib{1}$]
        ]
        [$\Fib{2}$, context subtree]
      ]
      [$\Fib{3}$, dotted spine edge]
    ]
  ]
]
\end{forest}
\caption{Recursive decomposition of $\Fib{2n+1}$ from Lemma~\ref{lem:oddconcat} underlying the construction in the proof of Theorem~\ref{thm:odd}. Dashed edges indicate the even-indexed Fibonacci word factors exposed by recursively expanding the odd-indexed Fibonacci words along the thick edges.}
\label{fig:tree-odd}
\end{figure}

Moving on to the odd-indexed case of Dumitran's conjecture, we see not only that we may once again focus entirely on prefix duplication, but also that we may reduce the bound on $k$ from $5$ to $3$.

\begin{theorem}\label{thm:odd}
$\Fib{2n+1} \in \PD_{k}^{*}(\Fib{2p+1})$ for all $1 \leq p \leq n$ and all $k \geq 3$.
\end{theorem}

\begin{proof}
If $p = n$, then the claim is immediate, since $\Fib{2n+1} \in \PD_{k}^{*}(\Fib{2n+1})$.

Now, assume that $p < n$. Applying Lemma~\ref{lem:oddconcat} with $i = n-p$ produces
\begin{equation*}
\Fib{2n+1} = \Fib{2n}\Fib{2n-2} \cdots \Fib{2p+2}\Fib{2p+1}.
\end{equation*}
Thus, it is enough to show that starting from $\Fib{2p+1}$, we may successively prepend the even-indexed Fibonacci words $\Fib{2p+2}$, $\Fib{2p+4}$, \dots, $\Fib{2n}$.

Since $2p+1 \geq 3$, Lemma~\ref{lem:F2prefix} says that $\Fib{2p+1}$ has prefix $\Fib{2}$. By our earlier observed consequence of Lemma~\ref{lem:evenPD}, whenever the current word $W$ begins with $\Fib{2}$, we may prepend any even-indexed Fibonacci word $\Fib{2r}$ via the derivation $W \Rightarrow_{k}^{*} \Fib{2r}W$. Applying this observation with $W = \Fib{2p+1}$ and $r = p+1$, we obtain
\begin{equation*}
\Fib{2p+1} \Rightarrow_{k}^{*} \Fib{2p+2}\Fib{2p+1}.
\end{equation*}
The resulting word begins with $\Fib{2}$, since $\Fib{2p+2}$ begins with $\Fib{2}$. Applying the same argument again with $r = p+2$, then with $r = p+3$, and so on, gives
\begin{equation*}
\Fib{2p+1} \Rightarrow_{k}^{*} \Fib{2p+2}\Fib{2p+1} \Rightarrow_{k}^{*} \Fib{2p+4}\Fib{2p+2}\Fib{2p+1} \Rightarrow_{k}^{*} \cdots \Rightarrow_{k}^{*} \Fib{2n}\Fib{2n-2} \cdots \Fib{2p+2}\Fib{2p+1}.
\end{equation*}
Observe that the word produced by this derivation sequence is $\Fib{2n+1}$. Hence, $\Fib{2n+1} \in \PD_{k}^{*}(\Fib{2p+1})$ as required.
\end{proof}

Again, since $\PD_{k}(w) \subseteq \PSD_{k}(w)$, Theorem~\ref{thm:odd} implies the odd-indexed case of Dumitran's conjecture.

%%%

\subsection{Optimality of the Bound \texorpdfstring{$k \geq 3$}{k >= 3}}\label{subsec:bound}

Given that the lower bound on the duplication length $k$ was successfully reduced from $5$ in the odd-indexed case of Dumitran's original conjecture to $3$ in the statement of Theorem~\ref{thm:odd}, one may wonder whether it is possible to lower the bound even further. In this section, we will show that the bound $k \geq 3$ cannot be lowered in the statements of either Theorem~\ref{thm:even} or Theorem~\ref{thm:odd}. It suffices for us to rule out the case where $k = 2$, since $\PSD_{1}^{*} \subseteq \PSD_{2}^{*}$.

To establish the optimality of our lower bound on $k$, we will use the prefix--suffix square reduction operation, introduced by Bottoni, Labella, and Mitrana~\cite{BottoniLabellaMitrana2017PrefixSuffixSquareReduction} and defined as follows:
\begin{equation*}
\PSSR(x) = \{uy \mid x = uuy, u \in \Sigma^{+}, y \in \Sigma^{*}\} \cup \{yu \mid x = yuu, u \in \Sigma^{+}, y \in \Sigma^{*}\}.
\end{equation*}
We can likewise define the $k$-prefix--suffix square reduction operation, denoted $\PSSR_{k}(x)$, in the usual way for any positive integer $k$.

If a word $y$ is obtained from $x$ by a bounded prefix--suffix duplication of length at most $2$, then $x$ can be obtained from $y$ by deleting one copy from a square prefix or square suffix $uu$, where $|u| \leq 2$. Hence, if $y \in \PSD_{2}^{*}(x)$, then $y$ must reduce to $x$ by repeatedly deleting square prefixes or square suffixes of length $1$ or $2$.

For the even-indexed case, take $p = 1$ and $n = 2$. Then $\Fib{2} = 010$ and $\Fib{4} = 01001010$. The word $\Fib{4}$ has no square prefix of length $1$ or $2$, and it has only one square suffix of length at most $2$: $1010 = (10)^2$. Deleting one copy of this suffix gives $\PSSR_{2}(\Fib{4}) = 010010$, but this resultant word has no square prefix or square suffix of length $1$ or $2$. Therefore, every possible $2$-prefix square reduction sequence from $\Fib{4}$ terminates at $010010$, not at $\Fib{2}$, and it follows that
\begin{equation*}
\Fib{4} \not\in \PSD_{2}^{*}(\Fib{2}).
\end{equation*}

For the odd-indexed case, again take $p = 1$ and $n = 2$. Then $\Fib{3} = 01001$ and $\Fib{5} = 0100101001001$. The word $\Fib{5}$ has no square prefix or square suffix of length $1$ or $2$: its first two symbols are $01$, its first four symbols are $0100$, its last two symbols are $01$, and its last four symbols are $1001$. Thus, $\Fib{5}$ is irreducible under all $2$-prefix--suffix square reductions. Since $|\Fib{5}| > |\Fib{3}|$, any derivation producing $\Fib{5}$ from $\Fib{3}$ must be nonempty, so it follows that
\begin{equation*}
\Fib{5} \not\in \PSD_{2}^{*}(\Fib{3}).
\end{equation*}

Consequently, neither $k = 1$ nor $k = 2$ suffice, and so the bound $k \geq 3$ is optimal.

%%%%%%

\section{Algorithmic Implementation}\label{sec:pseudocode}

The constructive nature of the proofs of Theorems~\ref{thm:even} and~\ref{thm:odd} lends itself naturally to an algorithmic implementation. Each of the derivations used in those proofs can be expanded into a sequence of individual prefix duplication operations, and in doing so, the argument can be reformulated as a recursive procedure that outputs a valid $\PD_{k}^{*}$ derivation from a seed Fibonacci word $\Fib{p}$ to a target Fibonacci word $\Fib{n}$.

In this section, we present pseudocode for a procedure \textproc{Derive} (Algorithm~\ref{alg:fw-derivation}), which computes a derivation symbolically by recording the length of the prefix duplicated at each step, rather than constructing every intermediate word directly. Together with the seed Fibonacci word $\Fib{p}$, the resulting sequence uniquely determines every word in the derivation.

In Algorithm~\ref{alg:fw-derivation}, the global sequence $S$ stores the lengths of the prefixes to be duplicated. The instruction \textproc{Emit}($d$) appends an integer $d$ to $S$, and this corresponds to duplicating the prefix of length $d$ of the current word at the current step of the derivation. (Note that, at any point in the computation, the ``current word" is understood to be the word obtained from $\Fib{p}$ by applying, in order, the prefix duplications whose lengths have currently been stored in $S$.)

The procedure \textproc{BuildEven}($j$) implements the inductive construction from Lemma~\ref{lem:evenPD}. If the current word has the form $\Fib{2}W$, then the sequence given by this procedure transforms it into $\Fib{2j}W$. The procedure \textproc{PrependEven}($j$) implements the consequence of Lemma~\ref{lem:evenPD} by first duplicating the leading prefix $\Fib{2}$ and then calling \textproc{BuildEven}($j$) to transform $\Fib{2}$ into $\Fib{2j}$. Similarly, the procedure \textproc{PrependOdd}($j$) implements the claim from the proof of Theorem~\ref{thm:even} by prepending $\Fib{2j+1}$ to any word beginning with $\Fib{2}$.

\begin{remark}
Algorithm~\ref{alg:fw-derivation} implements the derivations arising from the proofs of Theorems~\ref{thm:even} and~\ref{thm:odd}, but does not attempt to optimize or minimize the number of prefix duplications used in each derivation. Indeed, the procedure uses only prefix duplication lengths of $2$ and $3$, and is therefore valid for every $k \geq 3$.
\end{remark}

\begin{algorithm}[!ht]
\caption{Fibonacci word derivation via prefix duplication (PD)}
\label{alg:fw-derivation}
\begin{algorithmic}[1]
\Statex \textbf{Global:} $p$ \Comment{index of seed Fibonacci word}
\Statex \textbf{Global:} $k$ \Comment{maximum allowable PD prefix length}
\Statex \textbf{Global:} $S$ \Comment{sequence of emitted prefix lengths}

\Procedure{Emit}{$d$}
    \State append $d$ to $S$
\EndProcedure

\Procedure{BuildEven}{$j$}
    \Comment{inductive construction (Lemma~\ref{lem:evenPD})}
    \If{$j = 1$}
        \State \Return
    \EndIf
    \State \Call{BuildEven}{$j-1$}
    \State \Call{Emit}{$2$} \Comment{\textproc{Emit}$(2)$ $\rightarrow$ duplicate $\Fib{1} = 01$}
    \For{$i \gets 1$ \textbf{to} $j-1$}
        \State \Call{PrependEven}{$i$}
    \EndFor
\EndProcedure

\Procedure{PrependEven}{$j$}
    \Comment{$W \Rightarrow_{k}^{*} \Fib{2j}W$, $W$ starting with $\Fib{2}$ (Lemma~\ref{lem:evenPD})}
    \State \Call{Emit}{$3$} \Comment{\textproc{Emit}$(3)$ $\rightarrow$ duplicate $\Fib{2} = 010$}
    \State \Call{BuildEven}{$j$}
\EndProcedure

\Procedure{PrependOdd}{$j$}
    \Comment{$W \Rightarrow_{k}^{*} \Fib{2j+1}W$, $W$ starting with $\Fib{2}$ (claim in Thm.~\ref{thm:even})}
    \State \Call{Emit}{$2$} \Comment{\textproc{Emit}$(2)$ $\rightarrow$ duplicate $\Fib{1} = 01$}
    \For{$i \gets 1$ \textbf{to} $j$}
        \State \Call{PrependEven}{$i$}
    \EndFor
\EndProcedure

\Function{Derive}{$n$}
    \If{$p < 2$ \textbf{or} $n < p$ \textbf{or} $n \not\equiv p \pmod 2$ \textbf{or} $k < 3$}
        \State \Return fail
    \EndIf
    \State $S \gets ()$
    \If{$p$ is even}
        \For{$j \gets p/2$ \textbf{to} $n/2 - 1$}
            \State \Call{PrependOdd}{$j$}
        \EndFor
    \Else
        \For{$j \gets (p+1)/2$ \textbf{to} $(n-1)/2$}
            \State \Call{PrependEven}{$j$}
        \EndFor
    \EndIf
    \State \Return $S$
\EndFunction
\end{algorithmic}
\end{algorithm}

%%%

\subsection{Correctness}\label{subsec:correctness}

The correctness of Algorithm~\ref{alg:fw-derivation} follows directly from the proofs given in Section~\ref{sec:proof}. The procedure \textproc{BuildEven}($j$) recursively expands the inductive construction in Lemma~\ref{lem:evenPD}, thereby transforming a word of the form $\Fib{2}W$ into $\Fib{2j}W$. From this, it follows that \textproc{PrependEven}($j$) prepends $\Fib{2j}$ to any word beginning with $\Fib{2}$. The correctness of both \textproc{BuildEven}($j$) and \textproc{PrependEven}($j$) follows simultaneously by strong induction on $j$; in the computation of \textproc{BuildEven}($j$), each call made to \textproc{PrependEven}($i$) has $i < j$, thus using only cases already established by the induction hypothesis.

Likewise, the procedure \textproc{PrependOdd}($j$) first duplicates the prefix $\Fib{1}$ and then successively calls \textproc{PrependEven}(1),\allowbreak \ldots,\allowbreak \textproc{PrependEven}($j$); by Lemma~\ref{lem:oddconcat}, this prepends $\Fib{2j+1}$ to the current word, exactly as in the claim used in the proof of Theorem~\ref{thm:even}.

If $p$ and $n$ are even, then the main loop successively prepends the odd-indexed factors appearing in Lemma~\ref{lem:evenconcat} and produces $\Fib{n}$ from $\Fib{p}$. If $p$ and $n$ are odd, then the main loop successively prepends the even-indexed factors appearing in Lemma~\ref{lem:oddconcat} and again produces $\Fib{n}$ from $\Fib{p}$. The preceding invariants ensure that each call to \textproc{Emit}($2$) occurs when the current word begins with $\Fib{1}$, while each call to \textproc{Emit}($3$) occurs when the current word begins with $\Fib{2}$. Since both lengths are at most $k$, each emitted prefix length is valid for $k \geq 3$.

Altogether, the procedure \textproc{Derive} returns a sequence of prefix lengths witnessing that $\Fib{n} \in \PD_{k}^{*}(\Fib{p})$.

%%%

\subsection{Time and Space Complexity}\label{subsec:complexity}

It is well known that the lengths of the Fibonacci words satisfy the Fibonacci recurrence, with $|\Fib{0}| = 1$ and $|\Fib{1}| = 2$. If $\varphi = (1+\sqrt{5})/2$ denotes the golden ratio, then $|\Fib{n}| \in \Theta(\varphi^{n})$.

Suppose that the procedure \textproc{Derive} returns a sequence of $m$ prefix lengths, and let $\Delta_{p,n} = |\Fib{n}| - |\Fib{p}|$. Since every emitted prefix length corresponds to a duplication of a prefix of length $2$ or $3$, we have
\begin{equation*}
\left\lceil \frac{\Delta_{p,n}}{3} \right\rceil \leq m \leq \left\lfloor \frac{\Delta_{p,n}}{2} \right\rfloor.
\end{equation*}
For every nontrivial input instance, $p$ and $n$ have the same parity, so $p \leq n-2$ and therefore
\begin{equation*}
|\Fib{n-1}| = |\Fib{n}| - |\Fib{n-2}| \leq \Delta_{p,n} \leq |\Fib{n}|.
\end{equation*}
Thus, $m \in \Theta(|\Fib{n}|) = \Theta(\varphi^{n})$. Each procedure call and loop iteration performs constant work apart from its recursive calls, and the total number of such calls and iterations is $O(m)$. Moreover, any derivation from $\Fib{p}$ to $\Fib{n}$ must contain at least $\lceil \Delta_{p,n}/k \rceil$ operations; when $k$ is fixed, this gives $\Omega(|\Fib{n}|)$ operations. Assuming that appending to $S$ takes constant time, Algorithm~\ref{alg:fw-derivation} runs in $\Theta(|\Fib{n}|)$ time for every nontrivial input $p < n$. (When $p = n$, the algorithm runs in $O(1)$ time.)

The recursive calls made by the procedure \textproc{BuildEven} have depth $O(n)$, and so Algorithm~\ref{alg:fw-derivation} uses $O(n)$ auxiliary space. Storing the complete sequence $S$ requires an additional $\Theta(|\Fib{n}|)$ space for every nontrivial input, although this could be avoided by streaming the sequence.

\begin{remark}
The authors have made available an accompanying software repository\footnote{\url{https://github.com/flarelabstfx/Fibonacci-words_Dumitrans-Conjecture_PD}} containing an implementation of the derivation procedure. Unlike Algorithm~\ref{alg:fw-derivation}, which follows the constructions given in the proofs of Theorems~\ref{thm:even} and~\ref{thm:odd} directly, the implementation in the repository is optimized such that, given a bound $k$ and values $p$ and $n$ corresponding to a seed word $\Fib{p}$ and a target word $\Fib{n}$, respectively, it always returns a \emph{shortest possible} $\PD_{k}$-derivation of $\Fib{n}$.
\end{remark}

%%%%%%

\section*{Acknowledgements}\label{sec:acknowledgements}

The authors would like to acknowledge support received through the Mitacs Globalink Research Award program, which allowed the first author to complete this research while visiting St.\ Francis Xavier University. This research was supported by the Natural Sciences and Engineering Research Council of Canada (NSERC) Discovery Grant RGPIN-2024-04799.

%%%%%%

\bibliographystyle{plain}
\bibliography{References}
%\nocite{*}

%%%%%%

\end{document}